\documentclass[preprint,12pt,authoryear]{article}
\usepackage{amssymb,amsfonts,amsmath,amsthm}
\usepackage{graphicx}
\usepackage{caption}
\usepackage{graphicx, wrapfig, lipsum}
\usepackage{enumerate}
\usepackage{setspace}
\usepackage{array}
\usepackage{epsfig}
\usepackage{float}
\usepackage{subfigure}
\usepackage[margin=1in]{geometry}
\usepackage{amsmath}
\usepackage{hyperref}
\usepackage{epstopdf}
\usepackage{lineno}
\hypersetup{
	colorlinks,%
	citecolor=blue,%
	filecolor=red,%
	linkcolor=red,%
	urlcolor=blue
}
\numberwithin{equation}{section} \numberwithin{theorem}{section}

\def\bt{\begin{thm}}
	\def\et{\end{thm}}

\def\bl{\begin{lem}}
	\def\el{\end{lem}}

\def\bd{\begin{defi}}
	\def\ed{\end{defi}}

\def\bc{\begin{cor}}
	\def\ec{\end{cor}}

\def\bp{\begin{proof}}
	\def\ep{\end{proof}}

\def\br{\begin{rem}}
	\def\er{\end{rem}}

\newtheorem{thm}{Theorem}[section]
\newtheorem{lem}{Lemma}[section]
\newtheorem{defi}{Definition}[section]
\newtheorem{rem}{Remark}[section]
\newtheorem{cor}{Corollary}[section]

\newcommand{\be}{\begin{equation}}
\newcommand{\ee}{\end{equation}}

\reversemarginpar

			\newlength{\figurewidth}
			\newlength{\smallfigurewidth}
\begin{document}
	\title{A Mathematical Model for Tumor Cell Population Dynamics Based on Target Theory and Tumor Lifespan}
	\author{Amin Oroji$^{1}$, Shantia Yarahmadian$^{2}$, Sarkhosh Seddighi$^{3}$ and Mohd
		Omar$^{4}\ast$}
	\date{}
	\maketitle
\begin{abstract}
Radiation Therapy (XRT) is one of the most common cancer treatment methods. In this paper, a new mathematical model is proposed for the population dynamics of heterogeneous tumor cells following external beam radiation treatment. According to the Target Theory, the tumor population is divided into $m$ different subpopulations based on the diverse effects of ionizing radiation on human cells. A hybrid model consists of a system of differential equations with random variable coefficients representing the transition rates between subpopulations is proposed. This model is utilized to simulate the dynamics of cell subpopulations within a tumor. The model also describes the cell damage heterogeneity and the repair mechanism between two consecutive dose fractions. As such, a new definition of tumor lifespan based on population size is introduced. Finally, the stability of the system is studied by using the Gershgorin theorem. It is proven that the probability of target inactivity post radiation plays the most important role in the stability of the system. 
\end{abstract}

\section{Introduction} \label{intro}
It has been experimentally verified that the ionization
process initiated by radiating particles leads to lesions in
the cells \cite{C}. The negative effect on DNA
structure, makes lesions the most harmful consequence of
radiotherapy \cite{WK}, \cite{H}. Substantial progress has been made both in the classification and evaluation of XRT treatment planning through probabilistic modeling. The most well-known models are  Tumor Control Probability (TCP) \cite{ZM}, \cite{DH}, \cite{GN} and Normal Tissue Complication Probability (NTCP) \cite{L}, \cite{KAB}.
\\\\
There have been numerous advances in stochastic modeling of tumor response to radiation treatment, including the linear quadratic model
\cite{ZM}, \cite{F}, cell population dynamics models \cite{QS}, \cite{SHH}, \cite{GLGV}, mixed-effects behavioral models
\cite{BSK} and cell cycle models \cite{KBF}. However, many of these models are constructed to evaluate certain important features of XRT, but they do not incorporate biological tumor damage heterogeneity, which is
the focus of our study. We refer the reader to Michelson and Leith
\cite{ML} for further information on different types of heterogeneity.\\\\
The key concept in understanding
XRT biology is the target theory \cite{R}. 
A target is a
radio-sensitive site within a cell. Each cell contains a certain
number of targets, which may be deactivated after being
hit by radiation particles. Moreover, between two consecutive dose
fractions, each target may become active again following immune
system reaction \cite{T}. Despite the development of several complex interpretations of the target theory, the essential principle entails radiation-induced apoptosis of the organism on account of target(s) inactivation within the organism.
Although targets are considered functioning biological units \cite{N}, the number of
targets and their locations in an organism are not always clear. With regard to cell sensitivity, the majority of models usually assume that cell sensitivity is constant during radiation \cite{KBV1}, \cite{KBV2}, \cite{OMH}. The same assumption is also made for cell populations, in that the viability of a surviving cell is  similar to an irradiated cell, i.e., all cells are assumed to have the same survival probabilities. However, theses assumptions may not be completely accurate, as there is strong evidence that damaged cells are unable to resist radiation
\cite{KBV1}, \cite{KBV2}.
\\\\
The clinical significance of the intra-tumor heterogeneity of cell
phenotypes and cell damage is discussed in \cite{Gupta}, and \cite{DFL}.
As such, providing a definition of a suitable treatment duration is rather a clinical challenge, especially when considering therapeutic response variability. In this regard,
Keinj et al. developed a discrete-time Markov chain multinomial
model for tumor response \cite{KBV1}, which employs the
target theory. This model inspects the number of surviving cells in the tumor but does not consider the tumor lifespan
to be able to measure the tumor's response to treatment.\\\\
In this study, we model tumor population dynamics via a system of ordinary differential equations. Thereafter, we evaluate the transition rates using a Markov chain. The model is then applied to the special case of $m=3$, which is related to the effect of radiation on cells by dividing them into three subpopulations: cells with no effect ($x_0$), cells with single-strand break ($x_1$) and cells with double-strand breaks ($x_2$). In addition, we analyze the system's stability in this case as well as the system bifurcation with two parameters.
\\\\
The paper is organized as follows: Section (\ref{postulates}) introduces the general theory and preliminary findings. In section (\ref{model}), the tumor growth model is
discussed terms of a system of ordinary differential equations with Markov chain coefficients.
The model calibration is presented in section (\ref{calibration}). Thereafter, a new definition for tumor lifespan is proposed in section (\ref{lifespan-defi}). Three targets in each cell and the model parameters are employed to analytically investigate the obtained ODE system stability in
section (\ref{stability}). 
Finally, section (\ref{conclusion}) concludes the study.
\section{Modeling assumptions and framework}\label{postulates}
We have considred the following assumptions in our modeling framework:
\begin{enumerate}
	\item \label{p1} Cells have the same phenotype but they act independently.
	
	\item \label{p2} In the radiotherapy process, the magnitude of each
	dose fraction ($u_0$) is constant during treatment
	(i.e. $u_0=2~Gy$). The time lag between two consecutive dose fractions
	is 24 hours.
	\item \label{p3} Each cell consists of $m$ targets, which may be deactivated  with probability $q$ after each dose fraction. $\textbf{P}(i, j)$ represents the
	treatment probability matrix in the transition from $i$ to $j$
	inactive targets, i.e., deactivating $j$ targets when $i$ targets has been disabled before \cite{KBV1}. $\textbf{P}(i, j)$  is written as:
	
	\begin{equation}\label{P3}
	\displaystyle \textbf{P}(i,j)= \left\{ {{{{m-i} \choose {j-i}}
			q^{j-i}(1-q)^{m-j}} \atop {0} } \hskip 1 cm {{i \leq j} \atop {j <
			i}} \right.
	\end{equation}
	
	\item{}\label{p4} Each target may be revived
	with probability $r$. As described previously \cite{KBV1}, $\textbf{R}(i, j)$ represents the repair probability matrix in the transition from $i$ to $j$, as given by:
	\begin{equation}\label{R}
	\textbf{R}(i,j)= \left\{ {{{i \choose j}r^{i-j}(1-r)^j} \atop {0}}
	\hskip 1cm {{j\leq i < m} \atop { i < j}} \right.
	\end{equation}
	where $\textbf{R}(m,m)=1$ and  $\textbf{R}(m,j)=0$ for $m \ne j$.
	
	\item{}\label{p5} $x_i$ indicates the cell subpopulation with
	$i$ deactivated target(s), where $i=0,...,(m-1)$. For $i\neq j$, each cell can
	move from $x_i$ to $x_j$ with the constant time-independent transition rate of $\alpha(i,j)$.
	
	\item{}\label{p6} A cell will undergo radiation-induced apoptosis if all targets are deactivated.
	The cell death rate in subpopulation $x_i$ is considered as constant, $D_i$.
	
	\item{}\label{p7} Cells can reproduce if all targets become active. For
	simplicity, we assume that just before the repair mechanism
	acts, cells in subpopulation $x_0$ can give birth to new cells
	proportional to subpopulation $x_0$ with a constant rate of $\beta$. As such, each cell in subpopulation $x_0$ can
	divide into exactly two daughter cells with probability $\mu$
	or it can remain unchanged with probability $(1-\mu)$ between two consecutive dose fractions.%
\end{enumerate}
\section{Model derivation}\label{model}
As indicated in Fig.~(\ref{schem}), tumor dynamics is generally described as
the effect of radiotherapy on the different tumor cell subpopulations. The conservation law for subpopulation $x_i$'s, $i=0, \dots, m-1$ is written as follows. For $i=0$:
\begin{figure*}[h] \small
	\centering
	\epsfig{width=1\figurewidth,file=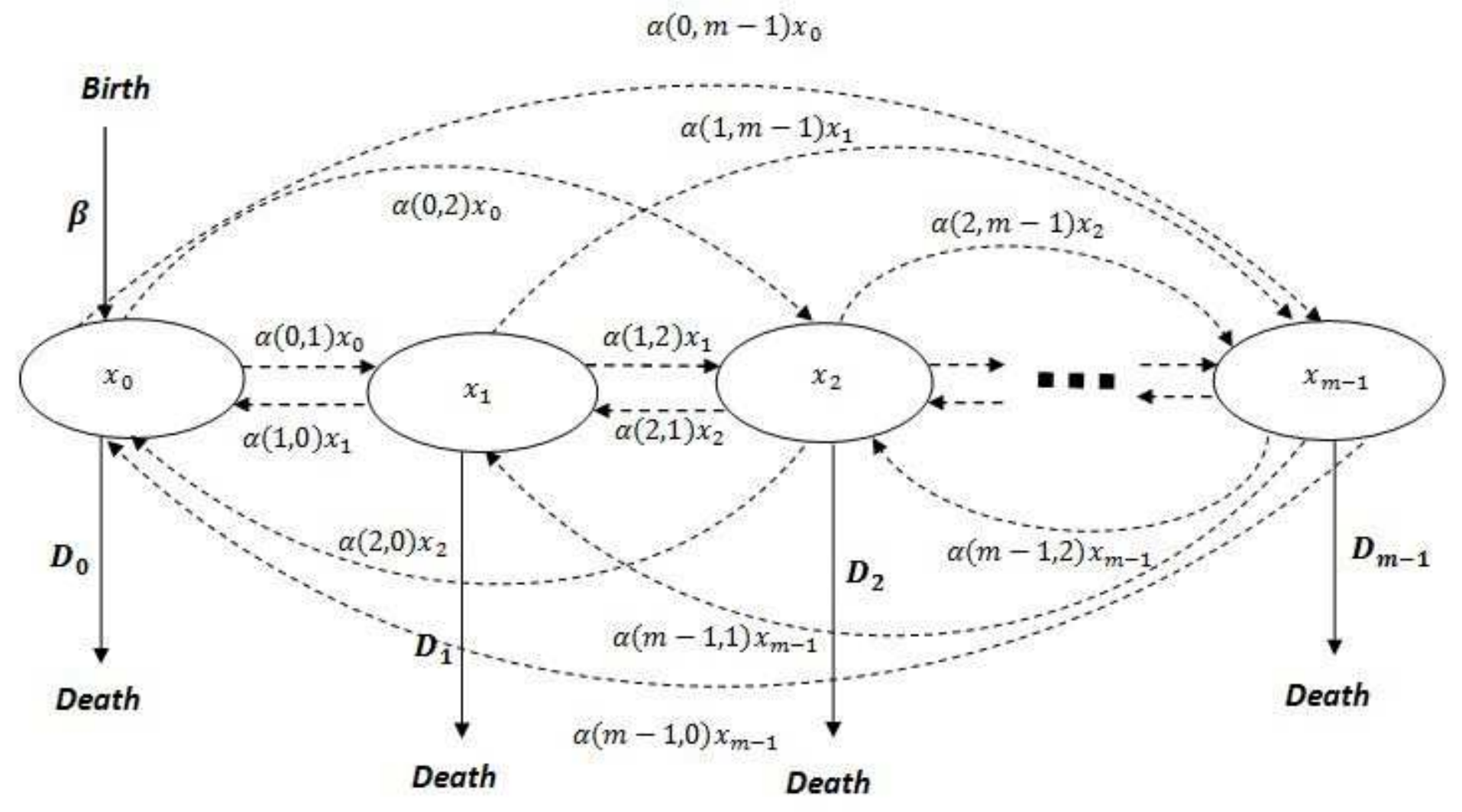}\\
	
	\caption{Schematic illustration of tumor cell
		population model.}
	\label{schem}
\end{figure*}

\begin{equation} \label{trans1}
{dx_0(t) \over dt}=\underbrace{\beta~ x_0(t)}_\text{Birth due to mitosis}+\underbrace{\sum_{j=1}^{m-1} \alpha(j,0)x_j(t)}_\text{transition to $x_0$}-\overbrace{\Big[\sum_{k=1}^{m-1} \alpha(0,k)\Big]~x_0(t)}^\text{transition from $x_0$}-\overbrace{D_0~x_0(t)}^\text{Death}
\end{equation}

and for $i\ne 0$:

\begin{equation} \label{trans2}
{dx_i(t) \over dt}=\underbrace{\displaystyle {\sum_{\underset{j\neq i}{j=0}}^{m-1}\alpha(j,i)~x_j(t)}}_\text{transition to $x_i$}-\overbrace{\Big[\sum_{\underset{k\neq 1}{k=0}}^{m-1} \alpha(i,k)\Big]~x_i(t)}^\text{transition from $x_i$} - \overbrace{D_i~x_i(t)}^\text{Death}
\end{equation}

These equations produce the following system:
\begin{eqnarray} \label{ode model}
{dx_0(t) \over dt}& =& \beta ~x_0(t) + \sum_{j=1}^{m-1} \alpha(j,0) ~x_j(t) - \Big[\sum_{k=1}^{m-1} \alpha(0,k) + D_0\Big]~x_0(t)  \\
{dx_1(t) \over dt}& =& \displaystyle{\sum_{\underset{j\neq
			1}{j=0}}^{m-1} \alpha(j,1) ~x_j(t)} -
\displaystyle{\Big[\sum_{\underset{k\neq 1}{k=0}}^{m-1} \alpha(1,k) +
	D_1\Big] ~x_1(t)} \nonumber
\\
&\vdots& \nonumber \\ {dx_{m-1}(t) \over dt}& =&  \sum_{j=0}^{m-2}
\alpha(j,m-1) ~x_j(t)-\Big[\sum_{k=0}^{m-2} \alpha(m-1,k) +
D_{m-1}\Big]~x_{m-1}(t) \nonumber
\end{eqnarray}
\section{Model Calibration}\label{calibration}
The probability that a cell will remain in $x_0$ after radiation
is $\textbf{P}(0,0)$. Therefore, the average number of births in
one day after applying the $k^{th}$ dose fraction and just
before the $(k+1)^{th}$ dose fraction is equal to:
\begin{eqnarray}\label{birth rate}
n(k) &=& x_0(k) \hspace{0.2cm} \mu
\hspace{0.2cm}\textbf{P}(0,0)\nonumber
\\&=& x_0(k) \hspace{0.2cm}\mu(1-q)^m
\end{eqnarray}
As seen in Eq.~(\ref{birth rate}), the newborn cells'
population size is proportional to $x_0$. Therefore, the birth
rate can be taken as:
\begin{equation}
\beta=\mu(1-q)^m  \hspace{0.2cm}
\end{equation}
\\
\begin{lem}\label{lem-Pi}
	Suppose that a cell has $i$ deactivated target(s) just before the
	application of a dose fraction and $\Pi = \textbf{PR}$. After treatment and repair,
	\begin{enumerate}
		\item $\Pi(i, j)$ represents the probability that a cell with $i$ deactivated target(s) just before the application of a dose fraction has $j$ deactivated target(s) right before the
		application of the next dose fraction. 
		\item An average number of $x_i\Pi(i,j)$ cells move from $x_i$ into $x_j$.
		\item For fixed $m$ and $k\geq 1$, the map $i\rightarrow\Pi^{k}(i,m)$ is increasing.
	\end{enumerate}
\end{lem}
\begin{proof}
	\begin{enumerate}
		\item Suppose that $\Pi=\textbf{PR}$. Therefore:
		\begin{equation}
		\Pi(i,j)=\sum_{k=0}^m{\textbf{P}(i,k)\textbf{R}(k,j)}
		\end{equation}
		Eq.~\eqref{R} shows that $\textbf{R}(m,j)=0$ for $j<m$. Therefore:
		\begin{equation}
		\Pi(i,j)=\sum_{k=0}^{m-1}{\textbf{P}(i,k)\textbf{R}(k,j)}
		\end{equation}
		
		Now assume that a cell has $i$ deactivated targets just before
		applying a dose fraction. After radiation and right before the
		repair mechanism, this cell may remain in subpopulation $x_i$ with
		probability $\textbf{P}(i, i)$, or it may move to subpopulation
		$x_k$, $k=i+1,...,(m-1)$, with probability $\textbf{P}(i,k)$. Following repair, this cell may move from subpopulation
		$x_k$ to subpopulation $x_j$ with probability $\textbf{R}(k,j)$.
		Therefore, the probability of transitioning from subpopulation
		$x_i$ into subpopulation $x_j$ after treatment and repair (one
		day) is $\Pi(i, j)$.
		\\
		\item The effect of treatment and repair on one cell is
		independent of the rest of the cells. Therefore, the average
		number of cells moving from subpopulation $x_i$ into subpopulation
		$x_j$ is equal to $x_i\Pi(i,j)$. \item See \cite{KBV2}.
	\end{enumerate}
\end{proof}

The following corollary is a direct consequence of lemma
\eqref{lem-Pi}.

\begin{cor}\label{cor-pi}
	With the same assumptions described in Lemma \eqref{lem-Pi}:
	\begin{enumerate}
		\item The cells' transition rate from subpopulation $x_i$ into
		subpopulation $x_j$ is equal to $\alpha(i,j)=\Pi(i,j) \hspace{0.2
			cm}(Day^{-1})$. 
		\item The death rate of subpopulation $x_i$ is $D_i = \Pi(i,m)\hspace{0.2cm}(Day^{-1})$.
	\end{enumerate}
\end{cor}
\begin{rem}\label{cor1}
	According to lemma \ref{lem-Pi}, we can separate the tumor
	cells into different sub-populations according to their sensitivity to
	the radiation. Therefore, the death rate in a subpopulation with more
	deactivated targets is higher than a subpopulation with fewer deactivated targets, which can be interpreted as the treatment heterogeneity in the model.
\end{rem}
Now, starting with subpopulation $x_0$, cells give birth at a constant
rate of $\mu (1-q)^m \hspace{0.2 cm} (Day^{-1})$. In addition, cells
transit from subpopulation $x_i$ into subpopulation $x_0$ at a rate of
$\Pi(i,0)\hspace{0.2 cm} (Day^{-1})$. Conversely, cells move from
subpopulation $x_0$ into subpopulation $x_i$ at a rate of
$\Pi(0,i)\hspace{0.2 cm} (Day^{-1})$ or may die
at a rate of $\Pi(0,m) \hspace{0.2 cm} (Day^{-1})$, where $\Pi$ is the transition matrix. Hence, for $i = 0, ..., (m-1)$ we have:
\begin{equation} {\label{trans}}
\sum_{l=0}^{m} \Pi(l,i)=1
\end{equation}	
By substituting Eq. \eqref{trans} in Eq. \eqref{trans1} we get:
\begin{eqnarray}
\label{x0}
{dx_0(t) \over dt}& =& [\Pi(0,0)+\mu(1-q)^m-1]~ x_0(t) +
\sum_{l=1}^{m-1} { \Pi(l,0) x_l(t)}
\end{eqnarray}
\\
Same analysis shows that for $i=1,\dots,(m-1)$:
\begin{eqnarray}
{dx_i(t) \over dt}& =&  \displaystyle {\sum_{\underset{l\neq
			i}{l=0}}^{m-1} \Pi(l,i)~x_l(t)}-x_i(t) ~\displaystyle {\sum_{\underset{l\neq i}{l=0}}^{m} {\Pi(i,l)}}
\end{eqnarray}

By substituting Eq. \eqref{trans} in Eq. \eqref{trans2} we get:
\begin{eqnarray}
\label{xi}
{dx_i(t) \over dt}& =&  \displaystyle {\sum_{\underset{l\neq
			i}{l=0}}^{m-1} \Pi(l,i)~x_l(t)}-x_i ~[1-\Pi(i,i)]
\end{eqnarray}

Finally, by substituting Eq.~\eqref{x0} and Eq.~\eqref{xi} in Eq.~\eqref{ode model}, the tumor growth model is described by:
\begin{eqnarray} \label{eq.system1}
{dx_0(t) \over dt}& =& [\Pi(0,0)+\mu(1-q)^m-1]~ x_0(t) +
\sum_{l=1}^{m-1} { \Pi(l,0) x_l(t)}\\
{dx_1(t) \over dt}& =& [\Pi(1,1)-1]~ x_1(t) +
\displaystyle{\sum_{\underset{l \neq 1}{l=0}}^{m-1} { \Pi(l,1)
		x_l(t)}} \nonumber.
\\
&\vdots&\nonumber\\ {dx_{m-1}(t) \over dt}& =& [\Pi(m-1,m-1)-1]~
x_{m-1}(t)+ \sum_{l=0}^{m-2} { \Pi(l,m-1) x_l(t)} \nonumber.
\end{eqnarray}
with initial conditions $\mathrm{x}(0)=(n_0,0,...,0)^\top$.
\section{Tumor lifespan}
\label{lifespan-defi}
What dose magnitude is required to remove the tumor completely? A small number of cells may still remain after resection, that are not visible and detectable by MRI. Therefore, it is crucial to know how many dose fractions must be applied to eliminate the remaining cancerous cells.
\\
The tumor lifespan is defined as the minimum number of dose fractions required
to remove the entire tumor \cite{AMS}. Therefore,
based on tumor population dynamics, the tumor lifespan is defined
as:
\begin{equation}\label{lifespan}
L=\min\{\lfloor t \rfloor :~ \lfloor N(t) \rfloor=0\}
\end{equation}
where
\begin{equation}
N(t)=\sum_{l=0}^{m-1} {x_l(t)}+n(\lfloor t \rfloor)
\end{equation}
 \section{Stability Analysis}\label{stability}
 Suppose that $m$ is an arbitrary integer. System \eqref{eq.system1} can be written as:
 \begin{equation}\label{system2}
 \dot x(t)= A(q,r)~x(t)
 \end{equation}
 where matrix $A$ is described as:
 \begin{equation}\label{matrix-A}
 A_{tk}= \left\{
 \begin{array}{lll}
 \Pi(0,0)+\mu (1-q)^m-1 & i,j=0 \\
 \Pi(i,i)-1             & i=j \hspace{.5cm}and \hspace{.5 cm}i\neq 0\\
 \Pi(j,i) &  i \neq j
 \end{array}
 \right.
 \end{equation}
 where $t=i+1$ and $k=j+1$. Therefore
 $$A_{11}=\Pi(0,0)+\mu (1-q)^m-1$$
 and for $2 \leq t \leq m$
 \begin{eqnarray}
 A_{tt} &=& \Pi(t-1,t-1)-1\\ \nonumber
 &=& \Pi(i,i)-1
 \end{eqnarray}
 
     Let $A_{m\times m}=(A_{tk})$ be a complex matrix. For $t \in {\{1,..., m\}}$ let  $\displaystyle{R_t = \sum_{k\neq{t}} \left|A_{tk}\right|}$ denote the sum of the absolute
     values of the non-diagonal entries in the $t$-th row and
     $\displaystyle{D(A_{tt}, R_t)}$ be the closed disc centered at
     $A_{tt}$ with radius $R_t$, which is known as Gershgorin disc. Eigenvalue of  $A$  lies within at least one of the
     	Gershgorin discs  $\displaystyle{D(A_{tt},R_t)}$ (Gershgorin Theorem \cite{Varga}).

 \begin{lem}\label{sum}
 	Suppose that $B=A^\top$. If $R_t$ defines as
 	\begin{equation}
 	R_t=\sum_{k\neq t}{B_{tk}}
 	\end{equation}
 	then
	 \begin{enumerate} \item \begin{equation} R_t > 0
 		\end{equation} \item
 		\begin{eqnarray}
 		B_{tt} + R_t=\left\{
 		\begin{array}{ll}
 		\mu (1-q)^m-q^m & i = 0 \\
 		-q^{(m-i)}             & 1 \leq i \leq (m-1) \\
 		\end{array}
 		\right.
 		\end{eqnarray}
 	\end{enumerate}
 	\begin{proof}
 		\begin{enumerate}
 			\item According to \eqref{matrix-A}, for $1 \leq t
 			\leq m$
 			\begin{eqnarray}
 			R_t&=&\sum_{k\neq t}{B_{tk}} \nonumber \\
 			&=& \sum_{j\neq i} {\Pi(i,j)} \nonumber \\
 			&>& 0
 			\end{eqnarray}
 			
 			\item First consider that $B=A^\top$ and $t=1$. Therefore:
 			\begin{eqnarray}
 			\sum_{k=1}^{m}{B_{1k}} &=& B_{11} + \sum_{k=2}^m B_{1k} \nonumber
 			\\
 			&=& \Pi(0,0) + \mu (1-q)^m -1 \nonumber + \sum_{j=1}^{m-1} \Pi(0,j) \nonumber \\
 			&=& \Pi(0,0) + \mu (1-q)^m -1 + (1-\Pi(0,0)-\Pi(0,m)) \nonumber
 			\\
 			&=& \mu(1-q)^m - q^m
 			\end{eqnarray}
 			Moreover, for $2 \leq t \leq m$
 			\begin{eqnarray}
 			\sum_{k=1}^{m}{B_{tk}} &=& B_{tt} + \sum_{k\neq t}{B_{tk}} \nonumber \\
 			&=& (\Pi(i,i) - 1) + \sum_{j\neq i}{\Pi(i,j)} \nonumber \\
 			&=& - \Pi(i,m) \nonumber \\
 			&=& - q^{m-i}
 			\end{eqnarray}
 		\end{enumerate}
 	\end{proof}
 \end{lem}
 
 The main result of this section is written as follows:
 \begin{thm}\label{thm.m.general}
 	For any $m\geq 2$, $0 < \mu \leq 1$ and $0 < r < 1$, the system
 	$\mathbf {\dot x} = A(q,r) \mathbf x$ is stable at equilibrium
 	point $\mathbf{0}$, where $q > 0.5$.
 \end{thm}
 \begin{proof}
 	It is enough to show that all eigenvalues of matrix $A$
 	have negative real parts   . To provide this we will show that for any
 	$q > 0.5$, $0 < r < 1$ and $m\geq 2$ any point of Gershgorin
 	circles $D(A_{tt} ,R_t)$ have negative real part, where $1 \leq t
 	\leq m$. For this purpose we apply Gershgorin Theorem on matrix
 	$B=A^\top$. Based on Lemma \eqref{sum},
 	\begin{eqnarray}
 	B_{tt} + R_t=\left\{
 	\begin{array}{ll}
 	\mu (1-q)^m-q^m & i = 0 \\
 	-q^{(m-i)}             & 1 \leq i \leq (m-1) \\
 	\end{array}
 	\right.
 	\end{eqnarray}
 	Note that the function $q^m$ is an increasing function for $q>0$ and $m$ is an integer. Therefore, for $1-q < 0.5 < q$ we have:
 	\begin{equation}
 	\mu (1-q)^m < (1-q)^m < q^m
 	\end{equation}
 	where $0 \leq \mu \leq 1$. \\
 	Consequently,
 	$B_{tt}\in \mathbb{R}$ and $B_{tt} + R_t < 0$ where $q >
 	0.5$ (Figure \eqref{g-d}). This shows that the Gershgorin Circles belong to the left half
 	of real line. In addition, according to Gershgorin
 	Theorem, each eigenvalue of matrix $B$
 	belongs in one of Gershgorin discs.    
 	Therefore, each eigenvalue of
 	matrix $B$ has negative real part. 
 	Hence, every eigenvalue of matrix $A$ has negative real part. This completes the proof.
 	\begin{figure*}[ht] 
 		\centering
 		\epsfig{width=1.4\figurewidth,file=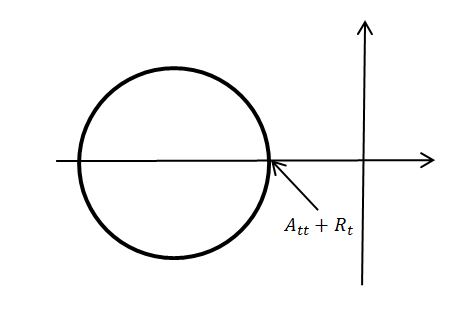}\\
 		\caption{\small{Gershgorin disc $B(A_{tt},R_t)$.}}
 		\label{g-d}
 	\end{figure*}
 \end{proof}

 \begin{thm}\label{thm.inf.general}
 	Suppose that $m\geq 2$ is an integer, $A \in \mathbb{M}^{m\times
 		m}$ and the set $S$ denotes the value $q$ such that the system \eqref{system2} is stable corresponding to all $0 < r < 1$ and $0
 	< \mu \leq 1$. Then:
 	\begin{equation}
 	\inf_q {A} = 0.5
 	\end{equation}
 \end{thm}
 \begin{proof}
 	According to Theorem \eqref{thm.m.general}, the
 	system \eqref{eq.system1} is stable for $q > 0.5$.
 	Now for any $\epsilon > 0$, corresponding to
 	$q_0=0.5-\epsilon$ and for $m=2$ there exists $r_0=1-\epsilon$
 	such that the system \eqref{eq.system1} is unstable. Hence:
 	\begin{equation}
 	\inf_q {A} = 0.5
 	\end{equation}
 \end{proof}
\section{Conclusion}\label{conclusion}
In this study, the population dynamics of tumor cells in the process of radiotherapy was examined. A system of differential equations with random variable coefficients was introduced to capture the heterogeneity of cell damage and the repair mechanism between two consecutive dose fractions. Subsequently, a new definition for tumor lifespan was introduced based on tumor population size. 
Based on the tumor lifespan, the effects of the probability that a target will be inactive after a dose fraction (q) and the probability that a target will reactivate after the repair mechanism (r) were investigated numerically. Our results are in good agreement with previously presented results \cite{KBV2}.
\section*{Acknowledgement}
The first author appreciates Dr. Ivy Chung and Dr. Ung Ngie Min from Faculty of Medicine, university of Malaya and Prof. Fazlul Sarkar from School of Medicine, Wayne State University for their constructing comments with regard to the manuscript. This study was financially supported by FRGS grant number FP015-2015A, University of Malaya.


\end{document}